\documentclass[conference]{IEEEtran}
\usepackage[numbers, square, comma, sort&compress]{natbib} 
\usepackage{graphicx}
\usepackage{color}
\usepackage{caption}
\usepackage[font=footnotesize]{subfig}
\usepackage{fixltx2e}
\usepackage{amsmath} 
\usepackage{amssymb,amsthm} 
\usepackage{units}
\usepackage{nicefrac}
\usepackage{multirow}
\usepackage{subfig}
\usepackage{hhline}

\newcommand{\argorder}{\operatornamewithlimits{argorder}} 

\usepackage{array} 

\newcommand{\definedas}{\overset{\underset{\Delta}{}}{=}}
\newcommand{\req}{\overset{\underset{!}{}}{=}}
\newcommand{\E}{\text{E}}

\newcommand{\e}{\text{e}}
\newcommand{\Q}{{\rm Q}}
\newcommand{\dd}{\text{d}}
\newcommand{\U}{\text{U}}
\newcommand{\pr}{\text{\text{Pr}}}

\newcommand{\Sa}{\mathcal{S}_{\text{a}}}
\newcommand{\Sac}{\mathcal{S}_{\text{a}}^{\complement}}
\newcommand{\argmin}{\operatornamewithlimits{argmin}} 
\theoremstyle{plain}
\newtheorem{theorem}{Theorem}
\theoremstyle{remark}
\newtheorem{remark}{Remark}

\usepackage{natbib}
\usepackage{url}

\begin{document}

\title{Multi-user Scheduling Schemes for Simultaneous Wireless Information and Power Transfer}

\author{\IEEEauthorblockN{Rania Morsi$\sp\dagger$, Diomidis S. Michalopoulos$^*$, and Robert Schober$\sp\dagger$}
\IEEEauthorblockA{$\sp\dagger$ Institute of Digital Communications, Friedrich-Alexander-University Erlangen-N\"urnberg (FAU), Germany\\
$^*$ Department of Electrical and Computer Engineering, University of British Columbia, Canada}
}
\maketitle

\begin{abstract}
In this paper,  we study the downlink multi-user scheduling problem for a time-slotted system with simultaneous wireless information and power transfer. In particular, in each time slot, a single user is scheduled to receive information, while the remaining users opportunistically harvest the ambient radio frequency (RF) energy.  We devise novel scheduling schemes in which the tradeoff between the users' ergodic capacities and their average amount of harvested energy can be controlled. To this end, we modify two fair scheduling schemes used in information-only transfer systems. First, proportionally fair maximum normalized signal-to-noise ratio (N-SNR) scheduling is modified by scheduling the user having the $j^{\text{th}}$ ascendingly ordered (rather than the maximum) N-SNR.  We refer to this  scheme as order-based N-SNR scheduling. Second, conventional equal-throughput (ET) fair scheduling is modified by scheduling the user having the minimum moving average throughput among the set of users whose N-SNR orders fall into a certain set of allowed orders $\Sa$ (rather than the set of all users). We refer to this scheme as order-based ET scheduling. The feasibility conditions required for the users to achieve ET with this scheme are also derived. We show that the smaller the selection order $j$ for the order-based N-SNR scheme, and the lower the orders in $\Sa$ for the order-based ET scheme,  the higher the average amount of energy harvested by the users at the expense of a reduction in their ergodic capacities. We analyze the performance of the considered scheduling schemes for independent and non-identically distributed (i.n.d.) Ricean fading channels, and provide closed-form results for the special case of i.n.d. Rayleigh fading. 
\end{abstract}

\IEEEpeerreviewmaketitle

\section{Introduction}
With the tremendous growth of the number of battery-powered wireless communication devices, the idea of prolonging their lifetime by allowing them to harvest energy from the environment has drawn a lot of research interest. Wireless power transfer is particularly important for energy-constrained wireless networks such as sensor networks. For such networks, replacing batteries can be costly or even infeasible in scenarios where the sensors are deployed in difficult-to-access environments or embedded inside human bodies. However, common renewable energy resources such as solar and wind energy are weather dependent and can not be used indoors. On the other hand, harvesting energy from radio frequency (RF) signals is a viable green solution to supply energy wirelessly to low-power devices \cite{powercast}. Moreover, RF signals can transport information together with energy, which motivates the integration of RF energy harvesting in wireless communication systems \cite{Varshney2008,Shannon_meets_tesla_Grover2010,MIMO_Broadcasting_Zhang2011}. 

Simultaneous wireless information and power transfer (SWIPT) systems were first studied in \cite{Varshney2008,Shannon_meets_tesla_Grover2010}, where the authors show that there exists a fundamental tradeoff between information rate and power transfer. This tradeoff can be characterized by the boundary of the so-called rate-energy (R-E) region \cite{MIMO_Broadcasting_Zhang2011}. However, due to practical circuit constraints, a received signal used for information decoding (ID) can not be reused for energy harvesting (EH) \cite{MIMO_Broadcasting_Zhang2011}. One possible practical receiver architecture for SWIPT is the time switching receiver which switches in time between EH and ID \cite{MIMO_Broadcasting_Zhang2011}. Recently, multi-user systems employing SWIPT have been  gaining growing attention in an effort to make SWIPT suitable for practical networks \cite{MIMO_Broadcasting_Zhang2011,Multiuser_MISO_beamforming2013,Throughput_Maximization_for_WPCN_Zhang2013}. Multi-user multiple input single output SWIPT systems are studied in \cite{Multiuser_MISO_beamforming2013}, where the authors derive the optimal beamforming design that maximizes the total energy harvested by the EH receivers under signal-to-interference-plus-noise ratio constraints at the ID receivers. Moreover, \cite{Throughput_Maximization_for_WPCN_Zhang2013} considers a multi-user time-division-multiple-access system with energy and information transfer in the downlink and uplink channels, respectively. The authors derive the optimal downlink and uplink time allocation with the objective of achieving maximum sum-throughput or equal-throughput under a total time constraint.\\
\indent Multi-user scheduling schemes that exploit the independent and time-varying multipath fading of the users' channels to create multi-user diversity (MUD) have been extensively studied in information-only transfer systems \cite{Unified_Scheduling_approach,Performance_Analysis_MUD_Alouini_2006}. With such schemes, the user having the most favourable channel conditions is opportunistically scheduled to transmit/receive. Practical opportunistic schedulers aim at maximizing the users' capacities, while maintaining long-term fairness  among users with different channel conditions. For example, scheduling the user having the maximum instantaneous normalized signal-to-noise ratio (N-SNR) (normalized to its own average SNR) maximizes the users' capacities while ensuring that the capacity of each user is proportional to his channel quality. Another way to provide fairness is to guarantee equal throughput (ET) to all users by scheduling the user having the minimum moving average throughput \cite{equal_throughput_2009}. In SWIPT systems, however, the scheduling rules should be modified to additionally control the amount of energy harvested by the users. Multi-user scheduling schemes that exploit MUD and guarantee long-term fairness among users have not been considered in the context of SWIPT so far. Thus, in this paper, we devise a proportionally fair \emph{order-based} N-SNR scheduler, where the user having the $j^{\text{th}}$-ordered N-SNR is scheduled and an \emph{order-based} ET fair scheduler, where the user having the minimum moving average throughput is scheduled among the set of users whose N-SNR orders fall into a set of allowed orders $\Sa$. Our results  show that the design parameters $j$ and $\Sa$ can be used to control the R-E tradeoff of the order-based N-SNR and ET schedulers, respectively.

\section{Preliminaries}
In this section, we present the considered system model and the EH receiver model. We also briefly review round robin (RR) scheduling which serves as a performance benchmark for the proposed order-based N-SNR scheduling scheme.
\subsection{System Model}
\label{s:system_model}
Consider a time-slotted SWIPT system that consists of one access point (AP) with a fixed power supply and $N$ user terminals (UTs) that are battery-powered. The system is studied for downlink transmission, where it is assumed that the AP always has a dedicated packet to transmit for every user. We assume that the UT receivers are time switching \cite{MIMO_Broadcasting_Zhang2011}, i.e., in each time slot, each UT may either decode the information from the received signal or harvest energy from it. In each time slot, the AP schedules one user for information transmission, while the other idle users opportunistically harvest energy from the received signal, as shown in Fig. \ref{fig:ID_EH_system}. Furthermore, the AP and the UTs are equipped with a single antenna. 
\begin{figure}
\centering
\includegraphics[width=0.4\textwidth]{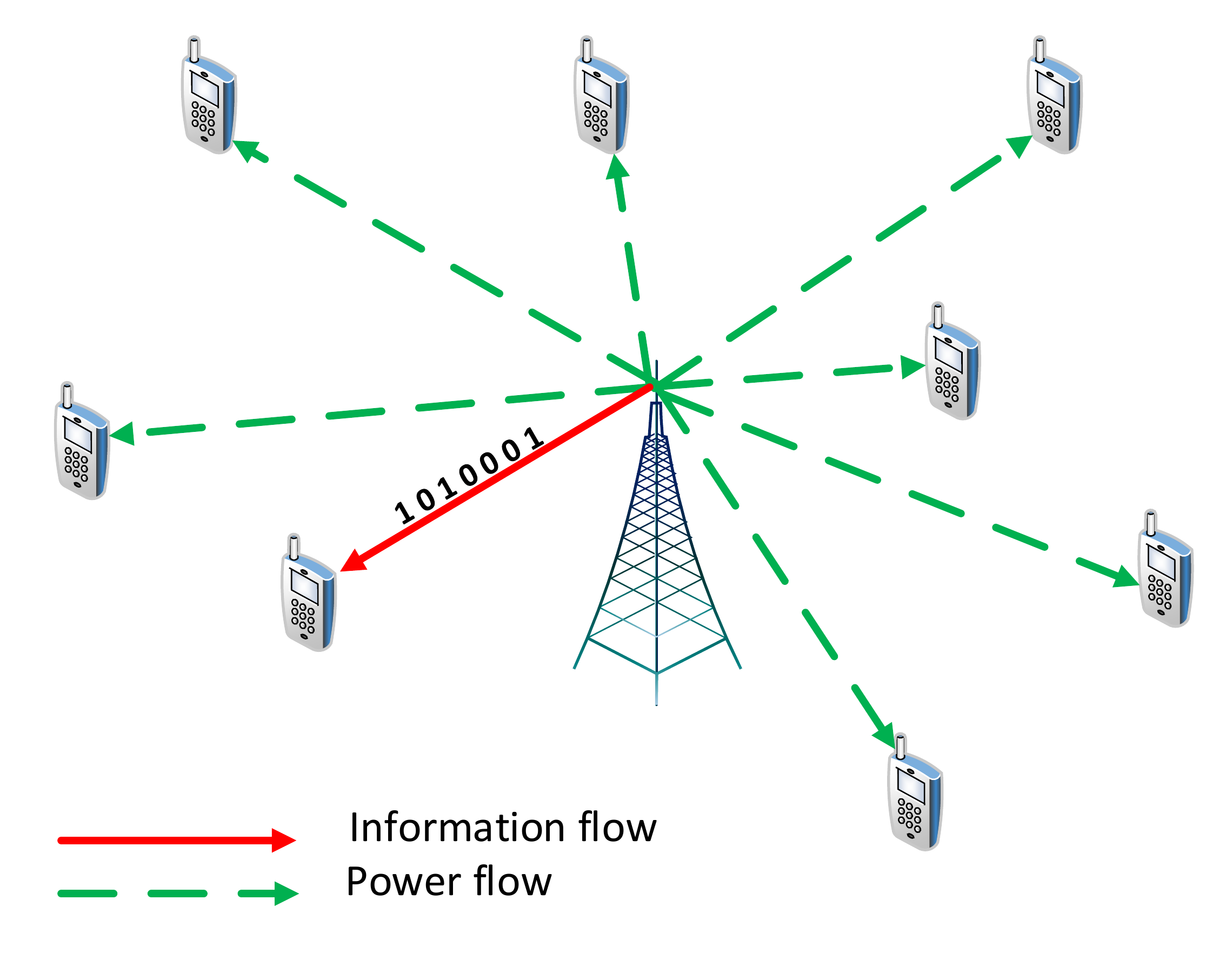}
\caption{Multi-user SWIPT system with time switching receivers. The scheduled user performs ID and the remaining users perform EH. }
\label{fig:ID_EH_system}
\end{figure}

At time slot $t$, the AP transmits an information signal to the scheduled user and the received signal at user $n$ is given by
\begin{equation}
r_n=\sqrt{P\,h_n}\e^{\text{j}\theta_n}x+z_n,
\label{eq:transmission}
\end{equation}
where $P$ is the constant transmit power of the AP, $x$ is the complex baseband information symbol whose average power is normalized to 1, i.e., $\E[|x|^2]\!\!=\!\!1$, where $\E[\cdot]$ denotes expectation, $\sqrt{h_n}$ and $\theta_n$ are respectively the amplitude and the phase of the fading coefficient of the channel from the AP to user $n$, and  $z_n$ is zero-mean complex additive white Gaussian noise with variance $\sigma^2$.\\
\indent The channels from the AP to the users are assumed to be block fading, i.e., the channel remains constant during one time slot, and varies independently from one slot to the next.
The  fading coefficients of the different links are assumed to be statistically independent.
Furthermore, the users' channels may have different mean channel power gains due to the path loss difference between users at different distances from the AP. We consider two channel models, namely Ricean and Rayleigh fading. The Ricean fading model is relevant for short-range RF EH applications, where a line of sight path may exist together with scattered paths between the transmitter and the receiver. In this case, $\sqrt{h_n}$ is Ricean distributed with Ricean factor $K_n$  and mean power gain $\Omega_n$, and the instantaneous channel power gain $h_n$ is non-central $\chi^2$-distributed with probability density function (pdf) \cite{Shannon_Capacity_fading_channels2005}%
\vspace{-0.2cm}
\begin{equation} 
f_{h_n}(x)=\frac{K_n+1}{\Omega_n} \e^{-K_n-\frac{(K_n+1)x}{\Omega_n}} \text{I}_0\left(2\sqrt{\frac{K_n(K_n+1)}{\Omega_n}x}\right),
 \label{eq:pdf_ncx2}
 \end{equation} 
where $\text{I}_0(\cdot)$ is the modified Bessel function of the $1^{\text{st}}$ kind and order zero. The corresponding cumulative distribution function (cdf) is \cite{Shannon_Capacity_fading_channels2005}%
\vspace{-0.3cm}
\begin{equation}
F_{h_n}(x)=1-\Q_1\left(\sqrt{2K_n},\sqrt{\frac{2(K_n+1)x}{\Omega_n}}\right),
 \label{eq:ind_Ricean_power_cdf}
 \end{equation}
where $Q_1(a,b)$ is the first-order Marcum Q-function given by $Q_1(a,b)=\int_b^\infty{x  e^{-\frac{(x^2+a^2)}{2}} \text{I}_{0}(ax) \dd x}$. If a direct link between the transmitter and the receiver does not exist, the Ricean factor is $K_n=0,\,\,\forall n\in\{1,\ldots,N\}$, and the distribution of the AP-UT channels reduces to Rayleigh fading. The resulting channel power gain of user $n$ is exponentially distributed with mean $\Omega_n$, pdf $f_{h_n}(x)=\lambda_n\e^{-\lambda_n x}$, where $\lambda_n=1/\Omega_n$, and cdf $F_{h_n}(x)=1-\e^{-\lambda_n x}$.
\subsection{Energy Harvesting Receiver Model}
We adopt the EH receiver model described in \cite{WIPT_Architecture_Rui_Zhang_2012}. In this model, the average harvested power (or equivalently energy, for a unit-length time slot) is given by \cite[Eq. (14)]{WIPT_Architecture_Rui_Zhang_2012}
\begin{equation}
EH=\eta h P,
\label{eq:EH}
\end{equation}
where $h$ is the channel power gain from the AP to the EH receiver and $\eta$ is the RF-to-DC conversion efficiency which ranges from 0 to 1. For currently commercially available RF energy harvesters, $\eta$ ranges from 0.5 to 0.7 \cite{powercast}.
\subsection{Round Robin (RR) Scheduling}
\label{ss:RR}
The RR scheduler grants the channels to the users in turn. Hence, no channel knowledge is needed at the AP. Therefore, a user receives information with probability $\frac{1}{N}$, and  harvests energy with probability $1-\frac{1}{N}$. Thus, user $n$ (denoted by $\U_n$) achieves an ergodic capacity $\E[C_{\U_n}]$ that is $\frac{1}{N}$ times the ergodic capacity $\E[C_{\text{U}_{n,\text{f}}}]$ that it would have achieved if it had \emph{full} time access to the channel. That is,  
\begin{equation}
\E[C_{\U_n}]\Big|_{\text{RR}}=\frac{1}{N} \E[C_{\text{U}_{n,\text{f}}}],
\label{eq:capacity_round_robin_nND_Ricean}
\end{equation}
where $\E[C_{\text{U}_{n,\text{f}}}]=\int_0^\infty{\log_2\left(1+P\frac{x}{\sigma^2}\right)f_{h_n}(x)\,\, \dd x}$ is given by \cite[Eq. (5)]{Shannon_Capacity_fading_channels2005} for Ricean fading. Note that we set the bandwidth in \cite[Eq. (5)]{Shannon_Capacity_fading_channels2005} to unity since $\E[C_{\text{U}_{n,\text{f}}}]$ is in bits/s/Hz. For Rayleigh fading,  the ergodic full-time-access capacity of user $n$ is $\E[C_{\text{U}_{n,\text{f}}}]=\frac{1}{\ln(2)}\e^{\frac{1}{\bar{\gamma}_n}}\,\E_1\left(\frac{1}{\bar{\gamma}_n}\right)$, where $\bar{\gamma}_n=\frac{P\Omega_n}{\sigma^2}$ is the average SNR of user $n$ and  $\E_1(x)$ is the exponential integral function of the first order defined as $\E_1(x)=\int_1^{\infty} \frac{\e^{-tx}}{t}\quad \dd t$.
The resulting average harvested energy of user $n$ for both Ricean and Rayleigh fading channels is 
\begin{equation}
\E\left[EH_{U_n}\right]=\left(1-1/N\right)\eta P \Omega_n.
\label{eq:EH_round_robin_IND_Ricean}
\end{equation}
\section{R-E tradeoff controllable Fair Scheduling Schemes For SWIPT Systems}
\label{s:Channel_aware_Scheduling_Schemes_For_Joint_information_and_Energy_Transfer}
Assuming full channel knowledge is available at the AP\footnote{The channel coefficients of the AP-UT links can be fed back from the users in the uplink in frequency division duplex systems  or can be assumed available in time division duplex systems due to channel reciprocity \cite{Unified_Scheduling_approach}.}, the scheduling decision can be made to control not only the capacity but also the amount of energy harvested by the users. In this section, we propose two fair scheduling algorithms in which the R-E tradeoff can be controlled.
\subsection{Order-based Normalized-SNR (N-SNR) Scheduling}
\label{s:order_based_NSNR_scheme}
In information-only transfer systems, the maximum N-SNR scheme schedules the user having the maximum instantaneous N-SNR, which  maximizes the users' capacities while maintaining proportional fairness among all users \cite{Performance_Analysis_MUD_Alouini_2006}. However, for SWIPT systems, such a scheduling rule leads to the minimum possible harvested energy by the users, since  the best states of the users' channels are exploited for ID rather than EH.  We propose to modify the maximum N-SNR scheme to an order-based scheme by scheduling the user having the $j^{\text{th}}$ ascendingly-ordered N-SNR, where the selection order $j$ is chosen from $\{1,\ldots,N\}$. If $j=N$, order-based N-SNR scheduling reduces to maximum N-SNR scheduling. 
\subsubsection{Order-based N-SNR Scheduling Algorithm}
Without loss of generality, ordering the N-SNRs in our model is equivalent to ordering the normalized channel power gains $\frac{h_n}{\Omega_n}$ since $P$ and $\sigma^2$ are identical for all users. Hence, the scheduling rule for the order-based N-SNR scheme reduces to
 \begin{equation}
n^*=\argorder\limits_{n\in\{1,\ldots,N\}} \frac{h_n}{\Omega_n}.
\label{eq:order_normalized_power_gain_rule}
\end{equation}
where we define \textquotedblleft$\argorder$" as the argument of the $j^{\text{th}}$ ascending order. 
\subsubsection{Performance Analysis} 
Next, we analyze the per user ergodic capacity and the per user average harvested energy of the considered scheduling scheme. For Ricean fading, the pdf of the metric to be ordered; the normalized channel power gain $X_n\!=\!\frac{h_n}{\Omega	_n}$, is given by (\ref{eq:pdf_ncx2}) but with unit mean (setting $\Omega_n\!=\!1$). Thus, $X_n,\,n=1,\ldots,N$, are independent and identically distributed (i.i.d.) if all users have the same Ricean factor, i.e., $K_n\!=\!K,$ $\forall n\!\in\!\{1,\ldots,N\}$. This is a realistic assumption since all users exist in the same physical environment. In this case, the pdf and cdf of the normalized Ricean fading channel power gains  $X_n,\,\forall n\!\in\!\{1,\ldots,N\}$ are $f_X(x)\!=\!(K+1) \e^{-K-(K+1)x} \text{I}_0\left(2\sqrt{K(K+1)x}\right)$ and $F_X(x)=1\!-\!\Q_1\left(\sqrt{2K},\sqrt{2(K+1)x}\right)$, respectively. The pdf of the $j^{\text{th}}$ ascendingly ordered random variable $X_{(j)}$ is given by \cite[eq. (2.1.6)]{Order_Statistics_David_Nagaraja}
\begin{equation}
f_{X_{(j)}}(x)\!=\!N\binom{N-1}{j-1}f_X(x)[F_X(x)]^{j-1}[1-F_X(x)]^{N-j}.
\label{eq:pdf_order_statistics}
\end{equation}
Thus, the ergodic capacity of user $n$ can be obtained as
\begin{equation}
\E\left[C_{j,\U_n}\right]=\frac{1}{N}\int\limits_0^\infty{\log_2\left(1+\bar{\gamma}_n x\right)f_{X_{(j)}}(x) \dd x},
\label{eq:Ergodic_capacity_Ui_IND_PFS}
\end{equation} 
where $\frac{1}{N}$ is the probability that the normalized channel of user $n$ has the $j^{\text{th}}$ order. The average harvested energy of user $n$ is  
\begin{align}
\E\left[EH_{j,\U_n}\right]&=\eta P \Omega_n\int\limits_{0}^{\infty}{x\left(f_{X}(x)- \frac{1}{N}f_{X_{(j)}}(x)\right) \dd x}\notag\\
&=\eta P \Omega_n\left[1-\E[X_{(j)}]/N\right],
\label{eq:avg_EH_NSNR}
\end{align}
where we used $f_X(x)=\frac{1}{N}\sum_{j=1}^N f_{X_{(j)}}(x)$. To the best of our knowledge, closed-form expressions for (\ref{eq:Ergodic_capacity_Ui_IND_PFS}) and (\ref{eq:avg_EH_NSNR}) do not exist for Ricean fading channels. Hence, we resort to numerical integration. For Rayleigh fading, we use (\ref{eq:pdf_order_statistics})--(\ref{eq:avg_EH_NSNR}) and the pdf and cdf of $X_n$, namely, $f_X(x)=\e^{-x}$ and $F_X(x)=1- \e^{-x}$, to obtain the per user ergodic capacity in closed-form as
\begin{align}
\E\left[C_{j,\U_n}\right]= \frac{\binom{N-1}{j-1}}{\ln(2)}\sum\limits_{l=0}^{j-1}&\frac{(-1)^l \binom{j-1}{l}}{(N-j+l+1)}\e^{\frac{(N-j+l+1)}{\bar{\gamma}_n}}\notag\\
&\E_1\left(\frac{N-j+l+1}{\bar{\gamma}_n}\right),
\label{eq:Capacity_NSNR_Rayleigh}
\end{align}
and the per user average harvested energy as
\begin{equation}
\E\left[EH_{j,\U_n}\right]=\eta P\Omega_n\left(1-\frac{1}{N}\sum\limits_{l=N-j+1}^{N}{\frac{1}{l}}\right).
\label{eq:EH_NSNR_Rayleigh}
\end{equation}
For the special case of $j=N$, the resulting average system capacity $\sum_{n=1}^N{\E\left[C_{j,\U_n}\right]}$ reduces to \cite[Eq. (44)]{Performance_Analysis_MUD_Alouini_2006}.
\subsection{Order-based ET Scheduling}
\label{ss:controllable_ET_scheduling}
Conventionally, ET fairness can be achieved by scheduling the user having the minimum moving average throughput \cite{equal_throughput_2009}. However, with such a scheduling rule, neither the resulting ET nor the amount of energy harvested by the users can be controlled. Hence, we devise a new scheme which trades the ET level for the amount of energy harvested by the users. 
\subsubsection{Order-based ET Scheduling Algorithm}
First, the users' instantaneous N-SNRs are ascendingly ordered, and then among the set of users whose N-SNR orders fall into the set of allowed orders $\Sa$, the AP schedules the one having the minimum moving average throughput. Thus, at time slot $t$, the scheduler selects user $n^*$ that satisfies 
 \begin{equation}
n^*=\argmin\limits_{O_n\in \Sa} r_n(t-1),
\label{eq:order_ET_controllable_rule}
\end{equation}
where $O_n\in\{1,\ldots,N\}$ is defined as the order of the instantaneous N-SNR of user $n$, and $r_n(t-1)$ is the throughput of user $n$ averaged over previous time slots up to slot $t-1$. The throughput of the users is updated recursively as
\begin{equation}
r_n(t)= \left\{ \begin{array}{ll}
(1-\beta)r_n(t-1)+\beta C_n(t) &\mbox{\hspace{-0.1cm}if user $n$ is scheduled} \\
(1-\beta)r_n(t-1) &\mbox{\hspace{-0.1cm}otherwise,}\\
\end{array} \right.
\label{eq:throughput_update}
\end{equation}
where $C_n(t)=\log_2\left(1+\frac{Ph_n(t)}{\sigma^2}\right)$ is the feasible rate of user $n$ in time slot $t$, and $\beta\in(0,1)$ is a smoothing factor selected to asymptotically vanish (e.g., $\beta=1/t$)\footnote{An asymptotically vanishing $\beta$ ensures convergence of the moving average throughput $r_n(t)$ to its ensemble average $\E[C_{\U_n}]$ since the considered fading process $h_n(t)$ is assumed to be stationary \cite{Unified_Scheduling_approach}.}. Confining the set $\Sa$ to low orders (e.g., $\Sa= \{1,\ldots,\lfloor\frac{N}{2}\rfloor$\}) leads potentially to a larger amount of harvested energy compared to conventional ET scheduling (which uses $\Sa= \{1,\ldots,N\}$), at the expense of a reduced ET. This is because a user from the set of low N-SNR users is scheduled for data reception and the users having relatively high N-SNRs are selected for EH. We note that depending on the choice of $\Sa$, ET scheduling may not always be feasible. This issue is investigated later in Theorem \ref{theo:feasibility_conditions} in detail.
\subsubsection{Performance Analysis}
Next, we analyze the ergodic capacity and the average harvested energy per user. The ergodic capacity of user $n$ can be formulated as
\begin{align}
\E[C_{\U_n}]&=\E[C_{\U_n}|O_n\in \Sa] \times \text{Pr}(O_n\in \Sa)\notag\\
&=\frac{\left|\Sa\right|}{N} \int\limits_{0}^\infty\log_2(1+\bar{\gamma}_n x) \left(\frac{1}{|\Sa|}\sum\limits_{j\in\Sa}f_{X_{(j)}}(x) \right) \dd x\notag\\
&\times \text{Pr}(\U_n|O_n\in\Sa)\notag\\
&=\sum\limits_{j\in\Sa}\E[C_{j,\U_n}]\Big|_{\text{N-SNR}}\text{Pr}(\U_n|O_n\in\Sa),
\end{align}
where $|\cdot|$ denotes the cardinality of a set, $\frac{1}{|\mathcal{S}_{\text{a}}|}f_{X_{(j)}}(x)$ is the likelihood function that the order of the normalized channel $x$ of user $n$ is $j$ given that $j\in\Sa$, $\text{Pr}(\U_n|O_n\in\Sa)$ is the probability that user $n$ is scheduled given that it has the chance to be scheduled, and $ \text{Pr}(O_n\in \Sa)=\frac{|\Sa|}{N}$ since the probability that a user's N-SNR takes any order from $\{1,\ldots,N\}$ is $1/N$. To write the average capacity of user $n$ in terms of its unconditioned probability of being scheduled $p_n \definedas\text{Pr}(n^*=n)$, we use
\begin{equation}
p_n=\text{Pr}(\U_n|O_n\in\Sa)\frac{|\Sa|}{N}.
\label{eq:pn_pn_conditioned}
\end{equation}
Hence, the average capacity of user $n$ reduces to
\begin{equation*}
\E[C_{\U_n}]=\frac{N}{|\Sa|}\sum\limits_{j\in\Sa}\E[C_{j,\U_n}]\Big|_{\text{N-SNR}}p_n\req r,\quad \forall n\in\{1,\ldots,N\},
\end{equation*}
where the average capacities of all users are forced to be equal to $r$, as required for ET transmission. This requires the scheduling probability of user $n$ to be 
\begin{equation}
p_n=\frac{r}{\frac{N}{|\Sa|}\sum\limits_{j\in\Sa}\E[C_{j,\U_n}]\Big|_{\text{N-SNR}}}.
\label{eq:pn_step1}
\end{equation}
Since $\sum_{n=1}^{N}{p_n}=1$ must hold, the resulting ET reduces to
\begin{equation}
r=\frac{1}{\frac{1}{N}\sum\limits_{n=1}^{N}{\frac{1}{\frac{1}{|\Sa|}\sum\limits_{j\in\Sa}\E[C_{j,\U_n}]\Big|_{\text{N-SNR}}}}}.
\label{eq:controllable_ET}
\end{equation}
In words, the equal throughput achieved by all users using the order-based ET scheme can be obtained by first determining the arithmetic mean of the order-based N-SNR capacities for the orders in $\Sa$ for each user, and then the harmonic mean of the resulting quantity for all users. This indicates that the user having the worst average channel will have a dominant effect on the resulting  ET. Using (\ref{eq:pn_step1}) and (\ref{eq:controllable_ET}), the set of scheduling probabilities required for all users to achieve ET reduces to 
\begin{equation}
p_n=\left(\sum\limits_{i=1}^{N}{\frac{\sum\limits_{j\in\Sa}\E[C_{j,\U_n}]\Big|_{\text{N-SNR}}}{\sum\limits_{j\in\Sa}\E[C_{j,\U_i}]\Big|_{\text{N-SNR}}}}\right)^{-1}, \forall n\in\{1,\ldots,N\}.
\label{eq:pn_controllable_ET}
\end{equation} 
As mentioned before, for certain combinations of $\Sa$ and $\Omega_n,\, n=1,\ldots,N$, the order-based ET scheduling algorithm may fail to provide all users with ET. In particular, the set of scheduling probabilities $p_n,\, n=1,\ldots,N$, in (\ref{eq:pn_controllable_ET}) required for the users to achieve ET may be infeasible. In the following theorem, we provide necessary and sufficient conditions for the ET-feasibility of the order-based ET scheduling algorithm. 
\begin{theorem}
The order-based ET scheduling is ET-feasible iff
\begin{align*}
p_n &\leq \frac{|\Sa|}{N}, &&\; \forall n\in\{1,\ldots,N\},\\
\sum\limits_{l=1}^{L} p_{k_{l}} &\leq \frac{\binom{N-1}{|\Sa|-1}L+\binom{L}{|\Sa|}(1-|\Sa|)}{\binom{N}{|\Sa|}},  &&\;\begin{aligned}&\forall (k_1,\ldots,k_{L})\in\mathcal{C}_{L},\\
&\forall L=|\Sa|,\ldots,N,\end{aligned}
\end{align*}
where $\mathcal{C}_{L}$ is the set of all $\binom{N}{L}$ combinations $(k_1,\ldots,k_{L})$ of $\{1,\ldots,N\}$.
\label{theo:feasibility_conditions}
\end{theorem}
\begin{proof} Please refer to the Appendix.  \end{proof}
\begin{remark}
It can be verified that the conventional ET scheme, where $|\Sa|\!=\!N$, is always ET-feasible. Also, the second condition is always satisfied for $L\!=\!N$ as it reduces to $\sum_{n=1}^{N} p_n\leq 1$ which is satisfied with equality by definition. For the case when $|\Sa|\!=\!1$, the scheme reduces to the order-based N-SNR scheme discussed in Section \ref{s:order_based_NSNR_scheme} which always achieves proportional fairness but not ET fairness. In conclusion, by properly selecting $\Sa$ with $|\Sa|>1$, the ET --if feasible-- can be traded for the amount of energy harvested by the users. \end{remark}
\begin{remark}
In most practical scenarios, the order-based ET scheduling algorithm is ET-feasible. ET-infeasibility occurs when the mean channel power gains $\Omega_n$ of the users differ by many orders of magnitude. For example, a scenario with 4 users having Rayleigh fading channels with $\Omega_n=1,\,1,\,10^{-10},$ and $10^{-10}$ and with $\Sa=\{3,4\}$ is feasible as the resulting $p_n=\{0.0884,0.0884,0.4116,0.4116\}$ satisfy the conditions in Theorem \ref{theo:feasibility_conditions}. In contrast, the same scenario but with $\Omega_n=1,\,1,\,10^{-11},$ and $10^{-11}$ is infeasible since the required scheduling probability set $p_n=\{0.0603,0.0603,0.4397,0.4397\}$ does not satisfy the second feasibility condition in Theorem \ref{theo:feasibility_conditions} for $L=|\Sa|=2$.
\end{remark}

\indent Next we calculate the average amount of harvested energy per user. Defining $\Sac$ as the complement of set $\Sa$, the average harvested energy of user $n$ is given by
\begin{align*}
&\E[EH_{\U_n}]=\E[EH_{\U_n}|O_n\in \Sac] \times \text{Pr}(O_n\in \Sac)&\\
&\quad\quad\quad\quad\quad+\E[EH_{\U_n}|O_n\in \Sa] \times \text{Pr}(O_n\in \Sa)&\\
&=\int\limits_0^\infty \eta P \Omega_n x \frac{1}{|\Sac|}\sum\limits_{j\in\Sac}f_{X_{(j)}}(x) \dd x \times \frac{|\Sac|}{N}\\
&\quad+\int\limits_0^\infty \eta P \Omega_n x \frac{1}{|\Sa|}\sum\limits_{j\in\Sa}f_{X_{(j)}}(x) \left(1-\frac{p_n N}{|\Sa|}\right) \dd x\times\frac{|\Sa|}{N},
\end{align*}
where from (\ref{eq:pn_pn_conditioned}), $\frac{p_n N}{|\Sa|}$ is the conditional probability that user $n$ is scheduled given that $O_n\in\Sa$. After simple manipulations, $\E[EH_{\U_n}]$ reduces to 
\begin{align}
\E[EH_{\U_n}]&=\eta P \Omega_n \left[ \frac{1}{N}\sum\limits_{j=1}^N\E[X_{(j)}]- \frac{p_n}{|\Sa|}\sum\limits_{j\in\Sa}\E[X_{(j)}]\right]\notag\\
&=\eta P \Omega_n \left[1-\frac{p_n}{|\Sa|}\sum\limits_{j\in\Sa}\E[X_{(j)}]\right],
\label{eq:avg_EH_controllable_ET}
\end{align}
where the unit term is obtained using $\sum_{j=1}^{N} X_{(j)}\!=\!\sum_{n=1}^{N} X_n$, thus $\sum_{j=1}^{N} \E[X_{(j)}]\!=\!\sum_{n=1}^{N} \E[X_n]\!=\!N$, since the normalized channels $X_n$ are unit-mean random variables $\forall n\!=\!1,\!\ldots\!,N$. The average user capacities and harvested energies for the order-based ET scheme in (\ref{eq:controllable_ET}) and (\ref{eq:avg_EH_controllable_ET}) can be obtained in closed-form if their order-based N-SNR counterparts in (\ref{eq:Ergodic_capacity_Ui_IND_PFS}) and (\ref{eq:avg_EH_NSNR}) are in closed-form, which is the case for e.g., Rayleigh fading channels.
\section{Simulation Results} 
\label{s:simulation_results}
\begin{figure}[!tp]
\centering
\includegraphics[ width=0.46\textwidth,trim= 0 0 0 0.6cm, clip]{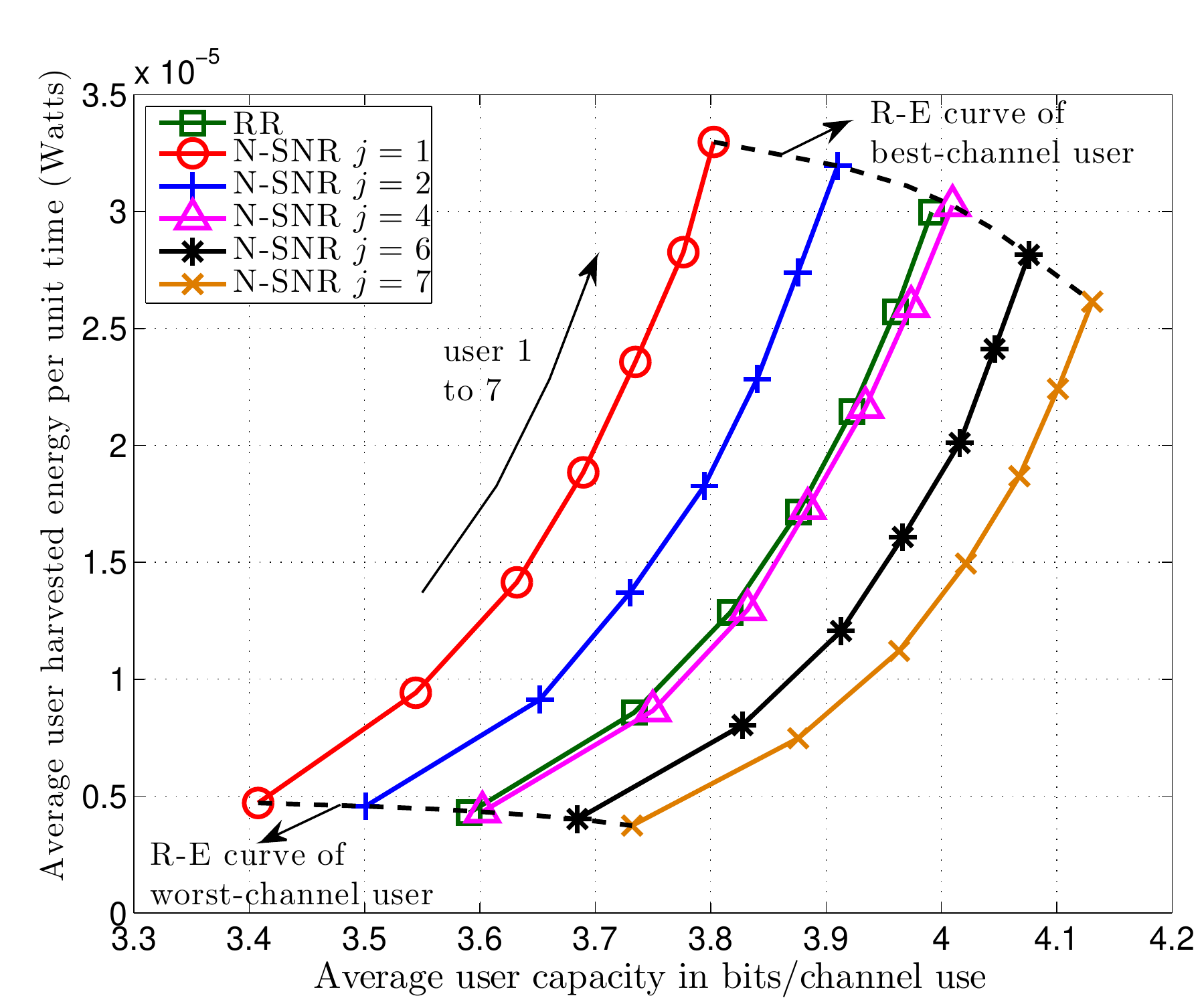} 
\caption{Rate-energy curves of the order-based N-SNR and the RR schemes for i.n.d. Ricean fading channels with $K=6$.}
\label{fig:Publication_curve_Rician_NSNR_ICC}
\end{figure}
\begin{figure}[!tp]
\centering 
\includegraphics[ width=0.5\textwidth,trim= 0 0 0 0.5cm, clip]{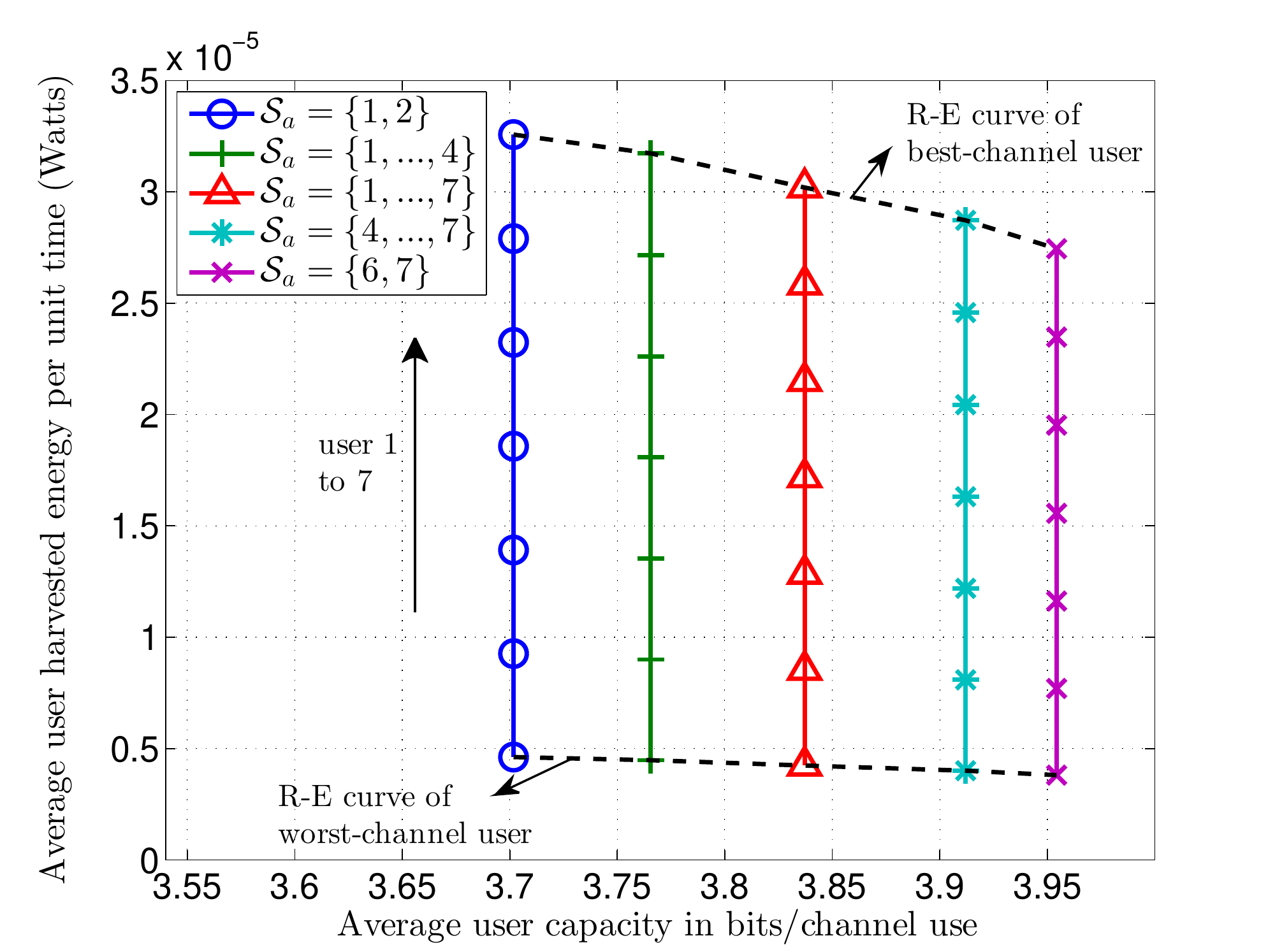} 
\caption{Rate-energy curves of the order-based ET scheduling scheme for i.n.d. Ricean fading channels with $K=6$.}
\label{fig:Publication_curve_Rician_ET_ICC}
\end{figure}
The investigated scheduling schemes have been simulated for an indoor environment with $N=7$ users operating in the ISM band at a center frequency of $\unit[915]{MHz}$, a bandwidth of $\unit[26]{MHz}$, and a noise power of $\sigma^2\!=\!\unit[-96]{dBm}$. We adopt the indoor path loss model in \cite{PathLoss_indoor_Rappaport_1992} for the case when the AP and the UTs are on the same floor (i.e., a path loss exponent of 2.76 is used, c.f. \cite[TABLE I]{PathLoss_indoor_Rappaport_1992}) and assume the AP-UT channels to be i.n.d. Ricean fading with Ricean factor $K\!=\!6$. We assume an AP transmit power of $P\!=\!\unit[1]{W}$, an antenna gain of $\unit[10]{dBi}$ at the AP and $\unit[2]{dBi}$ at the UTs, and an RF-to-DC conversion efficiency of $\eta\!=\!0.5$. The users' mean channel power gains are $\Omega_n\!=\!n\times10^{-5}$, which corresponds to an AP-UT  distance range of $\unit[2.27]{m}$ to $\unit[4.6]{m}$. The simulated results perfectly agree with the analytical results obtained based on the formulas reported in Section \ref{s:Channel_aware_Scheduling_Schemes_For_Joint_information_and_Energy_Transfer} for the considered scheduling schemes. Hence, only analytical results are provided in Figs. \ref{fig:Publication_curve_Rician_NSNR_ICC} and \ref{fig:Publication_curve_Rician_ET_ICC} for clarity of presentation . In Figs. \ref{fig:Publication_curve_Rician_NSNR_ICC} and \ref{fig:Publication_curve_Rician_ET_ICC},  every curve represents the ergodic capacity and average harvested energy of all users for a specific scheduling scheme. Also, the R-E curves of the worst and the best-channel users are highlighted. \\
\indent Fig. \ref{fig:Publication_curve_Rician_NSNR_ICC} shows the performance of the order-based N-SNR and the RR scheduling schemes. It is observed that both schemes achieve proportional fairness in terms of both the ergodic capacity and the average amount of harvested energy, since all users are on average scheduled the same number of times. Furthermore, RR scheduling performs  in-between the order-based N-SNR curves. This result was expected since the RR scheme is neither biased towards power transfer nor towards information transfer. Moreover, by reducing $j$ from $N$ to $1$, the order-based N-SNR scheduling allows the users to harvest more energy at the expense of reducing their ergodic capacities. For example, for the best-channel user reducing $j$ from $N$ to $1$ leads to a $7.94\%$ reduction in capacity and a $26.1\%$ increase in harvested energy.\\
\indent Fig. \ref{fig:Publication_curve_Rician_ET_ICC} shows the performance of the order-based ET scheduling scheme, which is feasible for all considered sets $\Sa$. The scheme yields ET for all users and an average harvested energy which is proportional to the users' channel conditions. It is observed that for the same $|\Sa|$, the higher the allowed orders in $\Sa$, the higher the ET at the expense of a reduction in harvested energy for all users. Hence, sets $\Sa\!=\!\{1,2\}$ and $\Sa\!=\!\{N-1,N\}$ provide the extreme ranges of such a tradeoff. In particular, going from $\Sa\!=\!\{6,7\}$ to $\Sa\!=\!\{1,2\}$ leads to an increase of $18.6\%$ and $21\%$ in the amount of harvested energy of the best and the worst channel users, respectively, at the expense of only a $6.33\%$ reduction in the ET.
\begin{remark}
We note that for lower data rates (capacities), the clock frequency and the supply voltage of the UT circuits can be scaled down (a technique known as dynamic voltage scaling). This leads to a cubic reduction in power consumption because dynamic power dissipation depends on the square of the supply voltage and linearly on the frequency ($P_c\propto V^2 f$) \cite{wang2009electronic}. Hence, when the users can tolerate low data rate, the selection order $j$ of the order-based N-SNR scheduling or the orders in $\Sa$ for the order-based ET scheduling can be chosen small to allow the users to harvest more RF energy and simultaneously reduce their power consumption. 
\end{remark}
\section{Conclusion}
\label{s:conclusion} 
This paper focused on modifying the scheduling objectives of energy-constrained multi-user communication systems by including the RF harvested energy as a performance measure. We presented novel order-based scheduling algorithms which provide proportional fairness/ET fairness and allow the control of the R-E tradeoff. We applied order statistics theory to analyze the per user ergodic capacity and average harvested energy for the considered schemes. Our results reveal that a smaller selection order for the order-based N-SNR scheme and lower orders in the selection set $\Sa$ for the order-based ET scheme result in a higher average harvested energy for all users at the expense of reduced average capacities.
\appendix
\label{app:feasibility}
The order-based ET scheduling may fail to provide all users with ET for one of the following reasons:
\begin{enumerate}
	\item Some user $n$ is required to be scheduled more often than possible. That is $\exists n: p_n>\frac{|\Sa|}{N} \overset{\text{from (\ref{eq:pn_pn_conditioned})}}{\equiv} \pr(\U_n|O_n\in\Sa)>1$. Thus, the first feasibility condition follows.
	\item For certain combinations of users in $\Sa$, the sum of the required probabilities that one of them accesses the channel exceeds one. That is, $\exists$ a combination $(k_1,\ldots,k_{|\Sa|})$ in $\{1,\ldots,N\}$, where $\sum_{l=1}^{|\Sa|}\pr(\U_{k_l}|\Sa=\{O_{k_1},\ldots,O_{k_{|\Sa|}}\})>1$.
\end{enumerate}
To find a simple condition for the second case, we first synthesize $\text{Pr}(\U_n|O_n\in\Sa)$ using the law of total probability as 
\small
\begin{equation*}
\text{Pr}(\U_n|O_n\!\in\!\Sa)\!=\!\frac{1}{\binom{N-1}{|\Sa|-1}}\sum\limits_{\mathcal{C}'_n}\pr(\U_n|\Sa\!=\!\{O_n,O_{i_1}\!,\!\ldots\!,O_{i_{|\Sa|-1}\!}\}), 
\end{equation*}
\normalsize
where $\mathcal{C}'_n$ is the set of all $\binom{N-1}{|\Sa|-1}$ combinations $(i_1,\ldots,i_{|\Sa|-1})$ from $\{1,\ldots\!,n\!-\!1\!,n\!+\!1,N\!\}$. Thus, from (\ref{eq:pn_pn_conditioned}),
\begin{equation}
\binom{N}{|\Sa|}p_n=\sum\limits_{\mathcal{C}'_n}\pr(\U_n|\Sa=\{O_n,O_{i_1},\ldots,O_{i_{|\Sa|-1}}\})
\label{eq:p_n_synthesized}
\end{equation}
holds $ \forall\, n\in\{1,\ldots,N\}$. In order to check that 
\begin{equation}
\sum\limits_{l=1}^{|\Sa|}\pr(\U_{k_l}|\Sa=\{O_{k_1},\ldots,O_{k_{|\Sa|}}\})=1
\label{eq:conditional_prob_must_1}
\end{equation}
holds $\forall$ combinations $(k_1,\ldots,k_{|\Sa|})$  drawn from  $\{1,\ldots,N\}$, we observe that adding $|\Sa|$ equations of (\ref{eq:p_n_synthesized}) for the users with indices $(k_1,\ldots,k_{|\Sa|})$ results in 
\begin{equation*}\binom{N}{|\Sa|}\sum\limits_{l=1}^{|\Sa|}p_{k_l}=\sum\limits_{l=1}^{|\Sa|}\pr(\U_{k_l}|\Sa=\{O_{k_1},\ldots,O_{k_{|\Sa|}}\})+\ldots.
\end{equation*}
Hence, applying (\ref{eq:conditional_prob_must_1}) and limiting every remaining probability term to $1$, the whole summation will be limited to \small
\begin{equation*}
\binom{N}{|\Sa|}\sum\limits_{l=1}^{|\Sa|} p_{k_{l}} \leq \hspace{-0.3cm}\underbrace{\binom{N-1}{|\Sa|-1}}_{\begin{aligned}&\text{total number of}\\[-0.5em] & \text{probabilities per}\\[-0.5em] & \text{eq. in (\ref{eq:p_n_synthesized}})\end{aligned}}\hspace{0.1cm}\underbrace{|\Sa|}_{\begin{aligned}&\text{number of} \\[-0.5em] &\text{eqs. added}\end{aligned}}+\underbrace{(1-|\Sa|)}_{\begin{aligned}&\text{replacing the number} \\[-0.5em] &\text{of probability terms}\\[-0.5em] &\text{that add to 1 by 1}\end{aligned}}.
\label{eq:feasibility_condition_Sa}
\end{equation*}\normalsize
Moreover, adding $L>|\Sa|$ equations of (\ref{eq:p_n_synthesized}) for the users with indices $(k_1,\ldots,k_L)$, and applying (\ref{eq:conditional_prob_must_1}) for every $|\Sa|$-length combination from the set $\{k_1,\ldots,k_L\}$, the whole summation will be limited to
\begin{equation*}
\binom{N}{|\Sa|}\sum\limits_{l=1}^{L} p_{k_{l}} \leq \binom{N-1}{|\Sa|-1}L+\binom{L}{|\Sa|}(1-|\Sa|).
\end{equation*}
which must hold for every combination $(k_1,\ldots,k_L)$ in $\{1,\ldots,N\}$ and for every $L=|\Sa|,\ldots,N$. Hence, the second feasibility condition follows.
\bibliographystyle{IEEEtran}
\bibliography{references}
\end{document}